\documentclass[]{revtex4-2}

\usepackage{amsmath}
\usepackage{amssymb}
\usepackage{array}
\usepackage{hhline}
\usepackage{hyperref}
\usepackage{amsthm}
\usepackage{enumerate}
\usepackage{array}
\usepackage{breqn}
\usepackage{tabularx}
\usepackage{physics}
\usepackage{tikz}
\usepackage{tabularx}
\usepackage{multirow}
\usepackage{todonotes}
\usepackage{mathtools}
\usepackage[margin=1in]{geometry}
\theoremstyle{plain}
\newtheorem{theorem}{Theorem}[section]

\newtheorem{lemma}[theorem]{Lemma}

\theoremstyle{definition}
\newtheorem{definition}[theorem]{Definition}

\newtheorem{conjecture}[theorem]{Conjecture}

\theoremstyle{remark}
\newtheorem{remark}[theorem]{Remark}

\newcolumntype{C}[1]{>{\centering\let\newline\\\arraybackslash\hspace{0pt}}m{#1}}

\begin{document}

\title{New Properties of Intrinsic Information and Their Relation to Bound Secrecy}

\author{Andrey Boris Khesin}
\affiliation{Department of Mathematics\char`,{} Massachusetts Institute of Technology\char`,{} \\
Cambridge\char`,{} Massachusetts\char`,{} 02139\char`,{} USA}
\author{Andrew Tung}
\affiliation{Menlo School\char`,{} \\
Atherton\char`,{} California\char`,{} 94027\char`,{} USA}
\author{Karthik Vedula}
\affiliation{James S. Rickards High\ School \\
Tallahassee\char`,{} Florida\char`,{} 32301\char`,{} USA}
\date{August 16, 2023}

\begin{abstract}
\noindent \textbf{Abstract}: 
The secret-key rate measures the rate at which Alice and Bob can extract secret bits from sampling a joint probability distribution, unknown to an eavesdropper Eve.
The secret-key rate has been bounded above by the intrinsic information and reduced intrinsic information.
However, we prove that the reduced intrinsic information is 0 if and only if the intrinsic information is 0.
This result implies that at least one of the following two conjectures is false: bound secrecy exists, or the reduced intrinsic information equals the secret-key rate.
We give an explicit construction of an information-erasing binarization for a candidate for bound secrecy. 
We then introduce some approaches for proving the existence of bound secrecy, such as reducing the channel space, linearly transforming Bob's map, and perturbing a channel for Eve.
\newline

\end{abstract}

\maketitle

\section{Introduction}
\label{intro}
A common problem in classical information theory is achieving secure communication over a public channel. Most modern-day cryptographic protocols rely on computational security, a type of security based on the computational difficulty of solving a certain problem. For example, the RSA protocol, widely used today, is based on the problem of factoring large integers \cite{mikeike, rsaCreation}. Unfortunately, the security of these types of protocols is always conditional because it relies on the fact that certain problems are computationally difficult, and that the adversary has limited computational power \cite{Maurer93}. Protocols based on information theory avoid this problem because the secrecy that they obtain is impossible for the eavesdropper to pierce, simply due to the laws of probability \cite{GiReWo02, shannon2}.

To achieve information-theoretic secure communication, most protocols begin with a procedure by which the two parties, call them Alice and Bob, agree on a secret key unknown to an eavesdropper Eve. Once this secret key is established, Alice and Bob can then encode an arbitrary message with the key completely securely. For example, suppose the secret key is composed of a string of bits. Then the message, in the form of another string of bits, can be perfectly secretly encoded by a one-time pad, which in this case can be performed by bitwise XOR \cite{xorCipher}. (Note that although it is perfectly secure, using a secret key as a one-time pad is not very efficient, and one often uses a cryptographic key expansion in cases where the secret key is expensive to generate \cite{oneTimePad}.)

Unfortunately for Alice and Bob, agreeing on an unconditionally secret key is impossible without a source of secrecy to start with \cite{maurer, shannon}. An example of such secrecy is if Alice and Bob could both observe the same random number generator, whose output is not available to an eavesdropper Eve. In this case, the amount of secrecy Alice and Bob share is simply the entropy of the random number generator, but in more complicated situations (e.g. if the output of the generator is partially known to Eve) secrecy is not as easy to quantify. Quantifying how much secrecy Alice and Bob share in a given situation has been attempted by introducing a number of quantities, such as the \textit{intrinsic information} and the \textit{reduced intrinsic information} \cite{GiReWo02, RenWol03}. A number of properties of these quantities have been discovered \cite{RenWol03, ReSkWo03}, suggesting that they are connected with the original problem of determining whether or not Alice and Bob can agree on a secret key (and if so, how long the key can be). For example, it has been proven that the intrinsic information is an upper bound on Alice and Bob's secret-key rate \cite{maurerwolf}.

However, some surprising results have shown that there is a gap between these information-theoretic quantities and Alice and Bob's ability to generate a secret key \cite{RenWol03}. A number of conjectures of this nature are currently unresolved, including the long-standing conjecture of the existence of bound secrecy. This conjecture has its origins in the analogous quantum phenomenon of bound entanglement, discovered in the late 1990s \cite{Horodecki_1998, GiReWo02, RenWol03, grw00, KhatriLutkenhaus}; the existence of bound secrecy was conjectured in the early 2000s. Bound secrecy refers to secrecy (i.e. positive intrinsic information) which cannot be extracted (i.e. the secret-key rate is $0$). If bound secrecy exists, it would suggest that classical information theory has surprising connections to quantum information theory, which was, in general, thought to be of a different nature. 

In this paper, we make an important step toward proving the existence of bound secrecy by showing that, in the crucial case where either intrinsic information or reduced intrinsic information is $0$, there is no gap between the two quantities. This is significant because the original purpose of introducing the reduced intrinsic information was to provide a stronger upper bound on the secret-key rate, and one of the prevailing approaches for constructing an example that has bound secrecy was showing the example has a positive intrinsic information but a reduced intrinsic information of 0 (implying that no secrecy can be extracted). This paper shows that this approach cannot work.

On the other hand, we suggest an alternative approach for establishing the existence of bound secrecy, first mentioned in \cite{GiReWo02}, based on the idea of binarizations. Binarizations are ways of processing a random variable stochastically such that the new random variable has two outputs. We show that the existence of bound secrecy can be reduced to a simple statement about binarizations and probability that must hold for $N$ copies of the distribution. We provide a proof of this statement for the case $N=1$ for a distribution introduced in \cite{GiReWo02}, and we suggest approaches to generalize the proof for larger values of $N$.

Additionally, we focus on a second family of distributions introduced in \cite{RenWol03} which are conjectured to be bound secret. Although the approach for the previous distribution does not completely carry over, we illustrate some possible approaches which promise to extend the methods used in the $N=1$ case of the previous distribution.

The outline of this paper is as follows. In Section \ref{background}, we formally define the secret-key rate, the intrinsic information, and the reduced intrinsic information, which will be important in the rest of the paper. We also give context for our result by summarizing the properties of these quantities which have been established previously. Additionally, we provide the formal statements of a number of important conjectures, such as the problem of bound secrecy, which are addressed in this paper. In Section \ref{mainresult}, we state and prove our results, which require a number of intermediate lemmas. In Section \ref{implications}, we discuss how our results relate to prior work, showing that given the existence of bound secrecy (which is widely believed to be true), another long-standing conjecture is false. In Section \ref{binarizations}, we discuss another approach to establish the existence of bound secrecy using binarizations of Alice's and Bob's random variables. In Section \ref{iib}, we improve on previous results by giving an explicit construction of a binarization which erases intrinsic information, which appears easier to generalize than previous non-constructive solutions. Finally, in Section \ref{family}, we provide multiple approaches and simplifications to prove the bound secrecy of another family of distributions, including reduction to Z-shaped channels, row-column-type transformations for weighted average target values, and isolated perturbations to show independence for a family of binarizations.

\section{Background}
\label{background}

The setup of the bound secrecy problem is as follows. Let $P_{XYZ}$ be a joint probability distribution of three countable (but possibly infinite) random variables $X$, $Y$, and $Z$, with Alice receiving $X$, Bob $Y$, and Eve $Z$. Throughout this paper we assume that any probability distribution is countable and has finite entropy. Entropy is denoted by $H$ and is assumed to be Shannon entropy; as such, all logs are assumed to be base $2$.

The \textit{secret-key rate} $S(X:Y || Z)$ is, informally, the rate at which Alice and Bob can extract secret bits from many copies of $P_{XYZ}$. The notation suggests the interpretation that the secret-key rate is the amount of information between $X$ and $Y$ given the information in $Z$. We are interested in the secret-key rate because if it is non-zero, Alice and Bob can extract their secret bits and thereby communicate securely. A formal definition of the secret-key rate, first introduced in \cite{maurer}, is as follows.

\begin{definition}
Suppose Alice and Bob are given $N$ independent realizations of a countable joint probability distribution $P_{XYZ}$. Call a protocol \textit{$\epsilon$-safe} if, at the end of the protocol, Alice and Bob can compute secret, correlated random variables $S_A$ and $S_B$ such that there exists another random variable $S$ so that 
\[P[S_A=S_B=S]>1-\epsilon \text{ and } I(S: CZ^N)<\epsilon.\]
Here, $C$ stands for any communications that took place during the protocol.
\end{definition}

The first condition ensures that Alice and Bob's variables must agree with probability very close to $1$, so that they share some information. The second condition ensures that this information is not accessible to Eve. This is defined formally using the \textit{mutual information} $I(X:Y) := H(X) + H(Y) - H(X,Y)$, a measure of the amount of information two random variables share. The condition requires that the mutual information between the secret variable $S$ and the pieces of data Eve has, namely $Z^N$ and the communications $C$, must be low.

Using the definition of an $\epsilon$-safe protocol, we define the secret-key rate asymptotically.

\begin{definition}
The \textit{secret-key rate} $S(X: Y || Z)$ is the largest number $R$ such that for all $\epsilon>0$, there exists an $N$ such that for all $n>N$, there exists an $\epsilon$-safe protocol using $n$ copies of $P_{XYZ}$ and producing the random variable $S$ with $\frac{H(S)}{N} \geq R$.
\end{definition}

Although the secret-key rate is the quantity we are interested in, as it captures the true number of bits Alice and Bob can extract, it has been hard to deal with because it allows any arbitrarily long communication string $C$. Ideally, one would express the secret-key rate $S(X: Y || Z)$ as a simple function of the distribution $P_{XYZ}$, but this problem is still open \cite{RenWol03}. Instead a number of upper bounds have been found. One of the first upper bounds on the secret-key rate was the conditional mutual information $I(X:Y|Z)$ \cite{blahut,maurerwolf}, defined as follows. 

\begin{definition}
Given a probability distribution $P_{XYZ}$, the \textit{conditional mutual information} $I(X:Y|Z)$ is defined as $H(X|Z) + H(Y|Z) - H(XY|Z)$, where each term is a conditional entropy conditioned on $Z$.
\end{definition}

One strategy for Eve to extract information about $X$ and $Y$ is to pass her variable $Z$ through a channel $P_{\overline{Z}|Z}$ \cite{ReSkWo03}, which in this case takes the form of a stochastic matrix acting on the vector of probabilities for $Z$. So we define the intrinsic conditional mutual information, first introduced in \cite{maurerwolf}.

\begin{definition}
Given a probability distribution $P_{XYZ}$, the \textit{intrinsic conditional mutual information} ~~~~ $I(X:Y \downarrow Z)$, sometimes called the \textit{intrinsic information}, is defined as
\[ I(X:Y \downarrow Z) := \inf_{P_{\overline{Z}|Z}} I(X:Y|\overline{Z}).\]
\end{definition}

\begin{theorem}[\cite{maurerwolf, RenWol03}]
Given a distribution $P_{XYZ}$, we have $S(X:Y||Z) \leq I(X:Y \downarrow Z)$. However, there exist distributions with $S(X:Y||Z) \neq I(X:Y \downarrow Z)$.
\end{theorem}

Motivated by the fact that $S(X:Y||ZU) \leq S(X:Y||Z) - H(U)$ holds but the corresponding inequality for the intrinsic information does not always hold, Renner and Wolf have introduced the \textit{reduced intrinsic conditional mutual information} \cite{ReSkWo03}.

\begin{definition}[\cite{RenWol03}]
Given a distribution $P_{XYZ}$, the \textit{reduced intrinsic conditional mutual information} $I(X:Y \downarrow \downarrow Z)$, sometimes called the \textit{reduced intrinsic information}, is defined as
\[ I(X:Y \downarrow \downarrow Z) := \inf_{P_{U|XYZ}} I(X:Y\downarrow ZU) + H(U). \]
\end{definition}

From the definition, we can see that the intrinsic information is an upper bound on the reduced intrinsic information, by setting $U$ to be trivial. The reduced intrinsic information is bounded from below by the secret-key rate, informally because in the infimum we can let $U$ be the secret bit that Alice and Bob can generate.

\begin{theorem}[\cite{RenWol03}]
Given a probability distribution $P_{XYZ}$, $S(X:Y||Z) \leq I(X:Y \downarrow \downarrow Z)$.
\end{theorem}

Intuitively, it is immediately obvious why the reduced intrinsic information should be different than the intrinsic information, because whatever information Eve receives through the variable $U$ is already accounted for in the $H(U)$ term. However it turns out that it is possible for the reduced intrinsic information to be strictly less than the intrinsic information because for some distributions Eve may have the additional disadvantage, beyond not knowing $X$ and $Y$, of not knowing how to process her variable $Z$. Therefore, the knowledge of how to process $Z$, as represented by $U$, can reduce the shared information between $X$ and $Y$ by more than the amount of information in $U$ itself. So the reduced intrinsic information is sometimes less than the intrinsic information.

\begin{theorem}[\cite{ReSkWo03}]
There exists a countable distribution $P_{XYZ}$ where $I(X:Y \downarrow Z) \neq I(X:Y \downarrow \downarrow Z)$.
\end{theorem}

As the reduced intrinsic information is a strictly stronger bound on the secret-key rate than the intrinsic information, it is natural to ask whether it in fact equals the secret-key rate. This open problem can be stated as follows.

\begin{conjecture}[\cite{RenWol03}]
\label{conjecture1}
Given a probability distribution $P_{XYZ}$, we have $S(X:Y||Z) = I(X:Y \downarrow \downarrow Z)$.
\end{conjecture}

Whereas previous bounds on $S$, such as the intrinsic information, have been improved by finding properties that were not shared between those quantities and $S$, so far the reduced intrinsic information appears to share many properties of the secret-key rate. If the conjecture is proven true (i.e. $S(X:Y||Z) = I(X:Y \downarrow \downarrow Z)$ in all cases), then we would have a relatively simple description, based on only the distribution $P_{XYZ}$, of the secret-key rate. This would fulfill one of the original objectives. If the conjecture is proven false, then it may reveal another potential strategy for Alice and Bob for secret-key extraction not related to intrinsic or reduced intrinsic information. Another significant conjecture is the problem of bound secrecy, namely secrecy between Alice and Bob that cannot be extracted.

\begin{conjecture}\cite{GiReWo02} (Bound secrecy)
\label{conjecture2}
There exists a distribution $P_{XYZ}$ such that $I(X:Y \downarrow Z) > 0$ but $S(X:Y||Z)=0$.
\end{conjecture}

This conjecture is inspired by the fact that a corresponding quantum phenomenon, bound entanglement, has been shown to exist \cite{Horodecki_1998, distillHorodecki}. Bound entangled states are quantum entangled states, analogous to classically correlated random variables, which have secrecy which cannot be distilled \cite{Horodecki_1998}. Relatively strong evidence suggesting the existence of bound secrecy has been found in \cite{GiReWo02, RenWol03} by drawing connections between the classical and quantum problems. Numerical evidence for bound secrecy has been given in \cite{KhatriLutkenhaus}.

\section{The Gap Between the Standard and Reduced Intrinsic Information}
\label{mainresult}

The main result of this paper is the following.

\begin{theorem}
\label{maintheorem}
Given a probability distribution $P_{XYZ}$, we have 
\[ I(X:Y \downarrow \downarrow Z) = 0 \iff I(X:Y \downarrow Z) = 0.\]
\end{theorem}

We first observe that the reverse direction follows because the intrinsic information is an upper bound on the reduced intrinsic information, which is nonnegative. We focus on the forward direction, whose proof takes the remainder of this section.

An important tool in the proof is the notion of the trace distance between two random variables.

\begin{definition}
Let two countable random variables $A$ and $B$ have probability distributions $\{a_i\}$ and $\{b_i\}$ with the same index set. Then the \textit{trace distance} between $A$ and $B$, denoted $D(A,B)$ is defined to be \[ D(A, B) := \frac{1}{2} \sum_i |a_i-b_i|. \]
\end{definition}

To prove the forward direction of Theorem \ref{maintheorem}, we reason as follows. If $I(X:Y\downarrow\downarrow Z)=0$, then by definition $\inf\limits_{P_{U|XYZ}}(I(X:Y\downarrow ZU)+H(U))=0$. First, suppose that this infimum is a minimum. This means that there exists an $XYZU$ such that $I(X:Y\downarrow ZU)+H(U)=0$, so $H(U)=0$ and $I(X:Y\downarrow ZU)=0$. However, since $U$ adds no information, we have $0=I(X:Y\downarrow ZU)=I(X:Y\downarrow Z)$, which is the desired statement.

From now on, assume that the infimum is not a minimum. This means that both quantities in the sum must approach 0 for a carefully chosen sequence of distributions. More rigorously, there must exist a sequence of probability distributions $\{XYZU_i\}$ such that $\lim\limits_{i\to \infty}H(U_i)=0$ and $\lim\limits_{i\to \infty}I(X:Y\downarrow ZU_i)=0$. Due to the definition of intrinsic information, there must also exist a sequence of channels $\{C_i\}$ such that $\lim\limits_{i\to \infty}I(X:Y|C_i(ZU_i))=0$.

In order to prove that $I(X:Y\downarrow Z)=0$, all that we have to do is show that there exists a sequence of channels $\{c_i\}$ such that $\lim\limits_{i\to\infty}I(X:Y|c_i(Z))=0$. We show this by showing $\{c_i\}=\{C_i\}$ works. In order to do this, we incorporate the defining property of the sequence $\{C_i\}$ by showing that \[\lim\limits_{i\to\infty}I(X:Y|C_i(Z))-I(X:Y|C_i(ZU_i))=0\] starting from $\lim\limits_{i\to\infty}H(U_i)=0$. In the rest of the proof, the channels $\{C_i\}$ are denoted using bars and the value of $i$ will be inferred from context; for example, we will write $\overline{Z}$ instead of $C_i(Z)$.

We first prove a number of lemmas regarding trace distances, denoted $D(A,B)$, and entropies. As a convention, let $K$ denote a constant random variable, whose probability distribution is a unit vector with the first component equal to $1$. The size of the range of $K$ is taken to be contextual (i.e. equal to the range of $U_i$).

Also, we assume that if $U_i$ is a random variable, then the probabilities for each outcome of $U_i$ are ordered in descending order. Such an ordering exists because any countable set of nonnegative values with total $1$ can be indexed in descending order: there can only be finitely many probabilities above any threshold $x \in (0,1)$, so we can order the probabilities that are above $x$ because there are only finitely many, and then order all the probabilities by repeatedly lowering $x$.

To prove Theorem \ref{maintheorem} we will prove the following sequence of implications.
\begin{align*}
        \lim_{i \to \infty} H(U_i) &= 0 \\
        \implies \lim_{i \to \infty} D(U_i, K) &= 0 \\
        \implies \lim_{i \to \infty} D(XYZU_i, XYZK) &= 0 \\
        \implies \lim_{i \to \infty} D(XY\overline{ZU_i}, XY\overline{ZK}) &= 0 \\
        \implies \lim_{i \to \infty} I(X:Y | \overline{ZU_i}) - I(X:Y|\overline{ZK}) &= 0
\end{align*}

To establish these implications, we first prove some lemmas.

\begin{lemma}
\label{lem3.2}
If $\lim\limits_{i \to \infty} H(U_i) = 0$ for some sequence of countable random variables $U_i$, then $\lim\limits_{i \to \infty} D(U_i, K) = 0$.
\end{lemma}

\begin{proof}
Suppose the probabilities for each outcome of the random variables $U_i$ are $a_{1i}$, $a_{2i}$, $\dots$, with $a_{1i} \geq a_{2i} \geq a_{3i} \geq \dots$. Then
\begin{align*}
    H(U_i) &= \sum_j a_{ji} \log \frac{1}{a_{ji}} \\
    &\geq \sum_j a_{ji} \log \frac{1}{a_{1i}} \\
    &= \log \frac{1}{a_{1i}}.
\end{align*}

Since $\log \frac{1}{a_{1i}}$ is nonnegative, if $H(U_i) \to 0$, we must have $\log \frac{1}{a_{1i}} \to 0$. Therefore $a_{1i} \to 1$.
If $k_1$, $k_2$, $\dots$ are the probabilities that $K = 1$, $K = 2$, and so on, then we have 
\[ D(U_i, K) = \frac{1}{2} \sum_{j} |a_{ji}-k_j| = \frac{1}{2} (1-a_{1i} + 1-a_{1i}) = 1-a_{1i} \]

so $D(U_i, K) \to 0$.
\end{proof}

\begin{lemma}
\label{lem3.3} 
Consider a sequence of countable random variables $U_i$ and let $Z$ be an arbitrary countable random variable. Then $D(ZU_i, ZK) = D(U_i, K)$.
\end{lemma}

\begin{proof}
For the purposes of this proof, let ``$1$" be the value that $K$ attains with probability $1$. It follows almost directly from the definition of trace distance that
\[D(ZU_i, ZK) = \sum_{(z,x)\in \mathcal{Z}\times\mathcal{U}_i} \max(0, P(Z=z, K=x) - P(Z=z, U_i=x)).\]
Since $K$ is always $1$, for any $x$ other than $1$, $P(Z=z, K=x)=0$, so $P(Z=z, K=x) - P(Z=z, U_i=x)\leq 0$ and $\max(0, P(Z=z, K=x) - P(Z=z, U_i=x))=0$. Thus the trace distance can be reduced to the following sum.
\[ D(ZU_i, ZK) = \sum_{(z,x)\in \mathcal{Z}\times\mathcal{U}_i} \max(0,P(Z=z, K=1) - P(Z=z, U_i=1)). \]
But since $K=1$ with probability $1$, and $P(U_i=1)\leq 1$, it suffices to take
\[ D(ZU_i, ZK) = \sum_{z \in \mathcal{Z}} P(Z=z) - P(Z=z, U_i=1) \]
which is just
\[D(ZU_i, ZK) = 1 - \sum_z P(Z=z, U_i=1) = 1-P(U_i=1). \]
As proven at the end of the proof for Lemma \ref{lem3.2}, $D(U_i, K) = 1-P(U_i=1)$ and we are done.
\end{proof}

\begin{remark}
The importance of $K$ is demonstrated by the above lemma, as the lemma becomes false if $U_i$ and $K$ are replaced by arbitrary random variables. A counterexample to Lemma \ref{lem3.3} in which $K$ is replaced by an arbitrary random variable is when $Z$ is a fair coin flip and $A=Z$ while $B$ is an independent fair coin flip. Then $D(A,B) = 0$ because these probability distributions are identical, but $D(ZA, ZB) = 1$ because $ZA$ is either both heads or both tails with probability $0.5$, while $ZB$ can be each of the 4 possibilities with probability $0.25$.
\end{remark}

\begin{remark}
The above lemma also shows the importance of converting statements about entropy into statements about trace distance (through Lemma \ref{lem3.2}) rather than some other measure of distance, such as the Kullback-Leibler (KL) divergence \cite{kl-divergence}. The KL divergence is defined for two probability distributions $P$ and $Q$, both over the probability space $\mathcal{X}$, as
\[D_{KL}(P||Q):=\sum_{x\in \mathcal{X}}P(x)\log\left(\frac{P(x)}{Q(x)}\right).\] 
Lemma \ref{lem3.3} does not make sense if the trace distances are replaced with KL-divergences because there exists a $Z$ with infinite range such that the KL-divergence of the left-hand side of the lemma $D_{KL}(ZU_i||ZK)$ diverges. Consider $P(Z=z_n)=2^{-n}$ and $P(Z=z_n,U_i=1)=2^{-n-\frac{2^n}{i}}$ for all $U_i$ (the rest of the $ZU_i$ probability distribution can be filled in arbitrarily). Here, as $i$ becomes larger, $P(U_i=1)$ becomes closer to 1, but
\[P(Z=z_n)\log\left(\frac{P(Z=z_n)}{P(Z=z_n,U_i=1)}\right)=\frac{1}{i}\]
so if $\mathcal{Z}$ is the range of $Z$,
\[D_{KL}(ZK,ZU_i)=\frac{|\mathcal{Z}|}{i}=\infty.\] 
\end{remark}

\begin{lemma}
\label{lem3.6}
Let $U_i$ be a sequence of random variables and let $Z$ be an arbitrary random variable. Suppose $C_i$ is a sequence of channels whose actions are denoted by a bar. Then for all $i$, $D(\overline{ZU_i}, \overline{ZK}) \leq D(ZU_i, ZK)$.
\end{lemma}

\begin{proof}
The proof is similar to that of the analogous quantum result, proven in \cite{mikeike}, that trace-preserving quantum operations are contractive.

For ease of writing let $X=ZU_i$ and $Y=ZK$. We can view the probability distributions $X$ and $Y$ as vectors of their probabilities ($\vec{x}$ and $\vec{y}$) and view the channel $C_i$ as a stochastic matrix which we denote $A$. We also let a subscript $i$ on a vector enclosed by parentheses (e.g. $(\vec{v})_i$) denote the $i$th component of the vector.

Using this notation, we have that
\[D(X,Y) = \sum_{i \text{ with } (\vec{x})_i-(\vec{y})_i>0} (\vec{x})_i-(\vec{y})_i = \sum_{i \text{ with } (\vec{x}-\vec{y})_i>0} (\vec{x}-\vec{y})_i. \]

Consider the vector $\vec{x}-\vec{y}$. We decompose this vector into its positive and negative components as follows. Let $\vec{a}$ be the vector defined by $\left(\vec{a}\right)_i = \max(0, \left(\vec{x}\right)_i-\left(\vec{y}\right)_i)$. Similarly, let $\vec{b}$ be the vector defined by $\left(\vec{b}\right)_i = \max(0, \left(\vec{y}\right)_i-\left(\vec{x}\right)_i)$. By definition, $(\vec{a})_i \geq 0$ for all $i$ and $\left(\vec{b}\right)_i \geq 0$ for all $i$. Therefore
\[D(X,Y) = \frac{1}{2} \left( \sum_i (\vec{a})_i + \sum_i (\vec{b})_i \right) .\]

We now prove the lemma:
\begin{align*}
    D(\overline{X}, \overline{Y}) &= \frac{1}{2} \sum_i |(A\vec{x})_i - (A\vec{y})_i| \\
    &\leq \frac{1}{2} \sum_i |(A\vec{a})_i| + |(A\vec{b})_i| \\
    &= \frac{1}{2} \sum_i (A\vec{a})_i + \frac{1}{2} \sum_i (A\vec{b})_i \\
    &= \frac{1}{2} \sum_i (\vec{a})_i + \frac{1}{2} \sum_i (\vec{b})_i \\
    &= D(X,Y)
\end{align*}
where the second to last step follows because the columns of $A$ sum to $1$ (as it is stochastic) and therefore $A$ preserves the sum of the elements of a vector.
\end{proof}

\begin{lemma}
\label{lem3.7} 
Given $P_{XYZ}$, we have that
$$\lim_{i \to \infty} D\left(XY\overline{ZU_i}, XY\overline{ZK}\right) = 0 \implies \lim_{i \to \infty} I\left(X:Y | \overline{ZU_i}\right) - I\left(X:Y|\overline{ZK}\right) = 0.$$
\end{lemma}

\begin{proof}
Define the following quantities:
\begin{itemize}
    \item Let $\mathcal{X}$ and $\mathcal{Y}$ denote the ranges of the random variables $X$ and $Y$, respectively. For any other variable $V$, let $\text{Range}(V)$ be the range of $V$.
    \item For all $x\in \mathcal{X}$, $y\in\mathcal{Y}$, $z\in\text{Range}\left(\overline{ZU_i}\right)$, we have $p_i(xyz):=p\left(X=x, Y=y, \overline{ZK}=z\right)$, and $q_i\left(xyz\right):=p\left(X=x, Y=y, \overline{ZU_i}=z\right)$.
    \item Let $Z^*_i:=\text{Range}\left(\overline{ZU_i}\right) \setminus \text{Range}\left(\overline{ZK}\right)$, and let $S_i:=\sum\limits_{\mathclap{XYZ_i^*}}q_i\left(xyz\right)$. Because $S_i\leq 2\cdot D\left(XY\overline{ZU_i}, XY\overline{ZK}\right)$, $S_i$ tends to 0 as $i$ tends to infinity.
\end{itemize}
Also, if there is a group of random variables (e.g. $X$, $Y$) or sets (e.g. $Z^*_i$) in the index of a summation, then the summation is iterated over all values in the range of each variable or element in the set, where the lowercase variables correspond to each of the uppercase random variables and sets (e.g. $x\in \mathcal{X}$, $y\in\mathcal{Y}$). Expanding the conditional mutual information expressions gives
\[I\left(X:Y|\overline{ZU_i}\right) - I\left(X:Y|\overline{ZK}\right)=\left(H\left(X\overline{ZU_i}\right)-H\left(X\overline{ZK}\right)\right)+\left(H\left(Y\overline{ZU_i}\right)-H\left(Y\overline{ZK}\right)\right)\]
\[-\left(H\left(XY\overline{ZU_i}\right)-H\left(XY\overline{ZK}\right)\right)-\left(H\left(\overline{ZU_i}\right)-H\left(\overline{ZK}\right)\right)=\]
\[=\sum_{\mathclap{XY\overline{ZU_i}}} \left(q_i\left(xyz\right)\log \left(q_i\left(z\right)\right)-p_i\left(xyz\right)\log \left( p_i\left(z\right)\right)\right)+\sum_{\mathclap{XY\overline{ZU_i}}}\left(q_i\left(xyz\right)\log\left( q_i\left(xyz\right)\right)-p_i\left(xyz\right)\log \left(p_i\left(xyz\right)\right)\right)\]
\[-\sum_{\mathclap{XY\overline{ZU_i}}} \left(q_i\left(xyz\right)\log \left(q_i\left(xz\right)\right)-p_i\left(xyz\right)\log \left(p_i\left(xz\right)\right)\right)-\sum_{\mathclap{XY\overline{ZU_i}}} \left(q_i\left(xyz\right)\log \left(q_i\left(yz\right)\right)-p_i\left(xyz\right)\log\left( p_i\left(yz\right)\right)\right).\]
Now, we split the summation into two parts: $z\in \text{Range}(\overline{ZK})$ or $z\in \text{Range}(Z^*_i)$. We now deal with the first part ($z\in \text{Range}(\overline{ZK})$). Note that
\[-H\left(XY\overline{ZK}\right)=\sum_{XY\overline{ZK}}p_i\left(xyz\right)\log \left(p_i\left(xyz\right)\right).\]
However, we also have that
\[-H\left(XY\overline{ZK}\right)=-H\left(XY\overline{ZU_i}|z\in\overline{ZK}\right)=\sum_{XY\overline{ZK}}\frac{q_i\left(xyz\right)}{1-S_i}\log\left(\frac{q_i\left(xyz\right)}{1-S_i}\right)=\]
\[-\log\left(1-S_i\right)+\frac{1}{1-S_i}\sum_{XY\overline{ZK}}q_i\left(xyz\right)\log \left(q_i\left(xyz\right)\right)\implies\]
\[\sum_{XY\overline{ZK}}q_i\left(xyz\right)\log \left(q_i\left(xyz\right)\right)=-\left(1-S_i\right)H\left(XY\overline{ZK}\right)+\left(1-S_i\right)\log\left(1-S_i\right).\]
This means that
\[\sum_{XY\overline{ZK}}\left(q_i\left(xyz\right)\log \left(q_i\left(xyz\right)\right)-p_i\left(xyz\right)\log \left(p_i\left(xyz\right)\right)\right)=S_iH\left(XY\overline{ZK}\right)+\left(1-S_i\right)\log\left(1-S_i\right).\]
This approaches 0 as $i$ goes to infinity because $S_i$ tends to 0. For the other summations, we can repeat this logic with $-H\left(X\overline{ZK}\right)$,  $-H\left(Y\overline{ZK}\right)$, and $-H\left(\overline{ZK}\right)$. This will produce the expressions
\begin{align*}
    &S_iH\left(X\overline{ZK}\right)+\left(1-S_i\right)\log\left(1-S_i\right), \\
    &S_iH\left(Y\overline{ZK}\right)+\left(1-S_i\right)\log\left(1-S_i\right), \\
    &S_iH\left(\overline{ZK}\right)+\left(1-S_i\right)\log\left(1-S_i\right),
\end{align*}
respectively. Therefore, for each of the four summations, the terms of the sum that are a part of $z\in \overline{ZK}$ approach 0. This deals with the part $z\in \overline{ZK}$.

Now, consider all $z\in \text{Range}(Z^*_i)$. Here, we have $p_i\left(\cdot, \cdot, z\right)=0$ because of the definition of $Z^*_i$. This leaves us with
\begin{align*}
    &\sum_{XYZ^*_i}q_i\left(xyz\right)\left(\log \left(q_i\left(z\right)\right)+\log \left(q_i\left(xyz\right)\right)-\log \left(q_i\left(xz\right)\right)-\log \left(q_i\left(yz\right)\right)\right) \\
    &=\sum_{XYZ^*_i}q_i\left(xyz\right)\log\left(\frac{q_i\left(z\right)}{q_i\left(xz\right)}\right)-\sum_{XYZ^*_i}q_i\left(xyz\right)\log\left(\frac{q_i\left(yz\right)}{q_i\left(xyz\right)}\right) \\
    &=\sum_{XZ^*_i}q_i\left(xz\right)\log\left(\frac{q_i\left(z\right)}{q_i\left(xz\right)}\right)-\sum_{XYZ^*_i}q_i\left(xyz\right)\log\left(\frac{q_i\left(yz\right)}{q_i\left(xyz\right)}\right).
\end{align*}
We show that both of these summations tend to 0. For the first summation, for all $x\in \mathcal{X}$, define
\[f\left(x\right):=\sum_{Z^*_i}q_i\left(xz\right)\log\left(\frac{q_i\left(z\right)}{q_i\left(xz\right)}\right).\]
Note that by the concavity of log, we have
\[f\left(x\right)=q_i\left(x\right)\sum_{Z^*_i}\frac{q_i\left(xz\right)}{q_i\left(x\right)}\log\left(\frac{q_i\left(z\right)}{q_i\left(xz\right)}\right)\leq q_i(x)\log\left(\sum_{Z_i^*}\frac{q_i\left(xz\right)}{q_i\left(x\right)}\cdot \frac{q_i\left(z\right)}{q_i\left(xz\right)}\right)= q_i\left(x\right)\log\left(\frac{S_i}{q_i\left(x\right)}\right).\]
This means that
\[0\leq\sum_{XZ^*_i}q_i\left(xz\right)\log\left(\frac{q_i\left(z\right)}{q_i\left(xz\right)}\right)=\sum_{X}f\left(x\right)\leq\sum_X q_i\left(x\right)\log\left(\frac{S_i}{q_i\left(x\right)}\right)=S_i\sum_X \frac{q_i\left(x\right)}{S_i}\log\left(\frac{S_i}{q_i\left(x\right)}\right)\]
\[=S_iH\left(X|z\in Z^*_i\right)\leq S_iH\left(X\right).\]

This means that the first summation tends to 0. The second summation also tends to 0 by replacing all instances of $z$ in the above proof with $yz$. Since all parts of the summations from the expanded conditional mutual information expressions tend to 0, we must have $I\left(X:Y|\overline{ZU_i}\right) - I\left(X:Y|\overline{ZK}\right)$ tends to 0 as well.
\end{proof}

We now prove Theorem \ref{maintheorem}.

\begin{proof}
We have the following sequence of implications, reproduced for clarity.
\begin{align*}
        \lim_{i \to \infty} H(U_i) &= 0 \\
        \implies \lim_{i \to \infty} D(U_i, K) &= 0 \\
        \implies \lim_{i \to \infty} D(XYZU_i, XYZK) &= 0 \\
        \implies \lim_{i \to \infty} D(XY\overline{ZU_i}, XY\overline{ZK}) &= 0 \\
        \implies \lim_{i \to \infty} I(X:Y | \overline{ZU_i}) - I(X:Y|\overline{ZK}) &= 0
\end{align*}
The first implication is a result of Lemma \ref{lem3.2}. Using Lemma \ref{lem3.3} and replacing $Z$ with $XYZ$ gives the second implication. Then, using Lemma \ref{lem3.6} with modified channels $\{C_i'\}$ that are identical to $\{C_i\}$, but they leave $X$ and $Y$ unchanged gives the third implication. Finally, using Lemma \ref{lem3.7} gives us the final implication.
\end{proof}

\section{Implications and Extensions}
\label{implications}
Theorem \ref{maintheorem} is a strengthening of a remark made by Christandl, Renner, and Wolf in \cite{ChReWo03}. In \cite{ChReWo03}, the authors prove that the infimum is a minimum in the definition of the intrinsic information as long as the range of $Z$ is finite. They remark that an argument analogous to that presented in their paper may prove that the infimum is a minimum in the definition of the reduced intrinsic information, with certain conditions on the size of $X$, $Y$, $Z$. This would imply a subcase of our theorem by the argument made briefly at the start of our proof. Unfortunately, it is unknown whether the arguments in \cite{ChReWo03} extend to the reduced intrinsic information measure.
However, our present result is stronger than results that might be obtained through these means because we only require $X$, $Y$, $Z$ to have finite entropy, whereas arguments analogous to those in \cite{ChReWo03} would require variables to have finite ranges.

Another application of Theorem \ref{maintheorem} is demonstrated in the following statement, mentioned briefly at the end Section \ref{intro}.

\begin{theorem}
If bound secrecy exists, then there exists a distribution $P_{XYZ}$ such that \[S(X:Y||Z) \neq I(X:Y\downarrow\downarrow Z).\]
\end{theorem}

\begin{proof}
Let $P_{XYZ}$ be a distribution that is bound secret, so that $I(X:Y\downarrow Z)>0$ and $S(X:Y||Z)=0$. However, by Theorem \ref{maintheorem}, we have $I(X:Y\downarrow Z)>0\implies I(X:Y\downarrow\downarrow Z)>0$. This means that this distribution satisfies $S(X:Y||Z)=0<I(X:Y\downarrow\downarrow Z)$, as desired.
\end{proof}

This theorem implies that at least one of the conjectures \ref{conjecture1} and \ref{conjecture2} is false. Since a significant amount of evidence suggesting the existence of bound secrecy has already been established, we believe that Conjecture \ref{conjecture1} is false. 

Furthermore, the above theorem implies that the approach of showing that a certain distribution is bound secret by computing a nonzero intrinsic information and a reduced intrinsic information of 0 is guaranteed to fail. In order for this approach to work, a property that would make Theorem \ref{maintheorem} false when the property is substituted for the reduced intrinsic information must be used. In particular, this property $f(XYZ)$ should satisfy the following: 
\begin{itemize}
    \item Given $P_{XYZ}$, we have $f(XYZ)\leq I(X:Y\downarrow Z)$, and equality does not always hold.
    \item $f(XYZ)=0$ does not imply $I(X:Y\downarrow Z)=0$.
\end{itemize}

\section{Binarizations and Bound Secrecy}
\label{binarizations}

One possible path for establishing the existence of bound secrecy has been suggested in \cite{binarizedSKR, GiReWo02}, which we now investigate. In \cite{GiReWo02}, the authors suggest that the existence of positive intrinsic information which vanishes upon binarization may be a candidate for bound secrecy. The authors provide an example of a distribution $X_0Y_0Z_0$ such that for all binarizations of $X_0$ and $Y_0$, producing $\overline{X_0}$ and $\overline{Y_0}$ respectively, $I(\overline{X_0}:\overline{Y_0} \downarrow Z_0) = 0$ (Proposition 4) \cite{GiReWo02}. They also show that for \textit{any} distribution $XYZ$, if the secret-key rate $S(X:Y||Z)$ is positive, then for some $N$ there exist binarizations of $X^N$ and $Y^N$ such that $I(\overline{X^N}: \overline{Y^N} \downarrow Z^N) > 0$ (Proposition 5) \cite{GiReWo02}. Therefore, the missing step for establishing bound secrecy for $X_0 Y_0 Z_0$ is the following:

\begin{conjecture}\cite{GiReWo02} \label{conjecture3}
Let $XYZ$ be a distribution. If, for all binary output channels $P_{\overline{X}|X}$ and $P_{\overline{Y}|Y}$ we have $I(\overline{X}: \overline{Y} \downarrow Z) = 0$, then for all $N$, for all binary output channels $P_{\overline{X^N}|X^N}$ and $P_{\overline{Y^N}|Y^N}$, we must have $I(\overline{X^N}: \overline{Y^N} \downarrow Z^N) = 0$.
\end{conjecture}

In fact, it is only necessary to prove Conjecture \ref{conjecture3} for the specific distribution $X_0 Y_0 Z_0$, which we will investigate in the next section. In this section, we reduce Conjecture \ref{conjecture3} to a much simpler statement which, if proven, would establish Conjecture \ref{conjecture3} and thereby prove the existence of bound secrecy. In the statement of the theorem, the symbol $\perp\!\!\!\perp$ is used to denote independence of random variables.

\begin{theorem}\label{326}
Conjecture \ref{conjecture3} is equivalent to the following:
\[ \forall \overline{X}, \overline{Y}, \exists \overline{Z} \text{ such that } (\overline{X} \perp\!\!\!\perp \overline{Y}) | \overline{Z} \implies \forall N, \forall \overline{X^N}, \overline{Y^N}, ~ \exists \overline{Z^N} \text{ such that } (\overline{X^N} \perp\!\!\!\perp \overline{Y^N}) | \overline{Z^N} \]
where the channels processing $X, Y, X^N, Y^N$ are assumed to be binarizations.
\end{theorem}

The statement of the theorem is noteable because it makes no reference to information-theoretic quantities: it is purely a statement about probabilities. To prove Theorem \ref{326}, we need the following lemma linking information and probability.

\begin{lemma}
Given random variables $X$, $Y$, $Z$, we have $I(X:Y|Z) = 0$ if and only if $(X \perp\!\!\!\perp Y) | Z$.
\end{lemma}
\begin{proof}
In this proof, for ease of writing we let $P(x)$ denote $P(X=x)$ for any $x \in \mathcal{X}$, and similarly for $Y$ and $Z$. To prove the forward direction, we have
\begin{align*}
    0 = -I(X:Y|Z) &= \sum_{xyz} P(x,y|z) \log \left( \frac{P(x|z) P(y|z)}{P(x, y | z)} \right) \\
    &\leq \frac{1}{\ln 2} \sum_{xyz} P(x, y|z) \left( \frac{P(x|z) P(y|z)}{P(x, y | z)} - 1\right) ~~~~~~~~~~~ \text{ since } \log_2 x \leq \frac{x-1}{\ln 2} \\
    &= \frac{1}{\ln 2} \sum_{xyz} \left(P(x|z) P(y|z) - P(x, y | z) \right) \\
    &= \frac{1}{\ln 2} \sum_{z} \left( \left(\sum_x P(x|z) \right) \left( \sum_y P(y|z) \right) - \sum_{xy} P(x,y|z) \right)
\end{align*}
Observe that for any $z \in \mathcal{Z}$, the sums $\sum\limits_x P(x|z)$, $\sum\limits_y P(y|z)$, and $\sum\limits_{xy} P(x,y|z)$ simply sum all the values in the conditional distribution $(XY)|Z=z$. So these sums all equal $1$, and the last line in the chain of expressions above is $0$. Since both sides of the above chain are $0$, the inequality must be an equality. Since $\log_2 x = \frac{x-1}{\ln 2}$ if and only if $x=1$, the expression inside the logarithm must always be $1$, which means \[P(X=x|Z=z) P(Y=y|Z=z) = P(X=x, Y=y|Z=z) \] for all $x, y, z$. Thus $(X \perp\!\!\!\perp Y) | Z$.

For the reverse direction, we simply note that \[ I(X:Y|Z) = \sum_{xyz} - P(X=x, Y=y|Z=z) \log \left( \frac{P(X=x|Z=z) P(Y=y|Z=z)}{P(X=x, Y=y | Z=z)} \right) \]
and if $(X \perp\!\!\!\perp Y) |Z$, then the expression inside the logarithm is always $1$, so each term of the sum becomes $0$, and $I(X:Y|Z) = 0$.
\end{proof}

We now prove Theorem \ref{326}.

\begin{proof}
As in the statement of the theorem, all channels that process $X, Y, X^N, Y^N$ are assumed to be binarizations. We observe that by the definition of the intrinsic information, \[ \forall \overline{X}, \overline{Y}, ~ I(\overline{X}:\overline{Y} \downarrow Z) = 0 \Longleftrightarrow \forall \overline{X}, \overline{Y}, \exists \overline{Z} \text{ such that } I(\overline{X}:\overline{Y} | \overline{Z}) = 0. \]
Then using the lemma, we have \[ \forall \overline{X}, \overline{Y}, \exists \overline{Z} \text{ such that } I(\overline{X}:\overline{Y} | \overline{Z}) = 0 \Longleftrightarrow \forall \overline{X}, \overline{Y}, \exists \overline{Z} \text{ such that } (\overline{X} \perp\!\!\!\perp \overline{Y}) | \overline{Z}. \]
We can repeat the logic for $X^N$, $Y^N$, $Z^N$. So \[\forall \overline{X}, \overline{Y}, ~ I(\overline{X}:\overline{Y} \downarrow Z) = 0 \implies \forall N, \forall \overline{X^N}, \overline{Y^N}, ~ I(\overline{X^N}:\overline{Y^N} \downarrow Z^N) = 0 \] is equivalent \[ \forall \overline{X}, \overline{Y}, \exists \overline{Z} \text{ such that } (\overline{X} \perp\!\!\!\perp \overline{Y}) | \overline{Z} \implies \forall N, \forall \overline{X^N}, \overline{Y^N}, ~ \exists \overline{Z^N} \text{ such that } (\overline{X^N} \perp\!\!\!\perp \overline{Y^N}) | \overline{Z^N} \] which is the desired result.
\end{proof}

Using Theorem \ref{326}, we can reduce the problem of bound secrecy to a statement simply about probability distributions and independence.

\section{Independence-Inducing Binarizations}
\label{iib}

We present some progress on proving Conjecture \ref{conjecture3} for the specific distribution $XYZ$, as introduced in \cite{GiReWo02}. This would be sufficient to establish bound secrecy for the distribution, as shown below. 
\begin{center}
\begin{tabular}{|l||c|c|c|}
    \hline 
    ~ $X$  & 1 & 2 & 3  \\ 
    $Y (Z)$ &&& \\
    \hhline{|=#=|=|=|}
    1 & 2 (0) & 4 (1) & 1 (2) \\ \hline
    2 & 1 (3) & 2 (0) & 4 (4) \\ \hline
    3 & 4 (5) & 1 (6) & 2 (0) \\
    \hline
\end{tabular}
\end{center}

For this distribution, the value of $Z$ is determined by the values of $X$ and $Y$, and is indicated by the number in parentheses in the cell. The unnormalized probability for that $xyz$ triplet is given by the number not in parentheses.

One method for proving the statement in Theorem \ref{326} for this distribution is by strengthening it to the following statement and not allowing Alice to binarize, as in the following conjecture.

\begin{conjecture}\label{strong326}
For the distribution $XYZ$, for any $N \geq 1$ we have
\[ \forall \overline{Y^N}, ~ \exists \overline{Z^N} \text{ such that } (X^N \perp\!\!\!\perp \overline{Y^N}) | \overline{Z^N} \]
where the channel processing $Y^N$ is assumed to be a binarization.
\end{conjecture}

Note that this conjecture implies Theorem \ref{326} because if Alice is not allowed to binarize and Eve can still erase correlation by processing $Z^N$, then there would still be no correlation even if Alice binarized her variable. To prove Conjecture \ref{strong326}, we must show that for any $N$, for any binarization that Bob chooses, Eve is able to process her variable such that Alice and Bob's variables are independent given Eve's information. Here, we primarily investigate the cases $N=1$ and $N=2$.

In the case $N=1$, it has been proven that for all binarizations $\overline{Y}$ of $Y$, Eve can always find a $\overline{Z}$ such that $X \perp\!\!\!\perp \overline{Y} | \overline{Z}$ (Proposition 4 of \cite{GiReWo02}). We have found an explicit construction of the map $\overline{Z}$, based on the following value.

\begin{definition}
We define the \textit{independence target value} (ITV) $\tau(x,y,z)$ for any $x,y,z \in \mathbb{R}$ as the median of $\frac{2x+1y+0z}{3}$, $\frac{1x+0y+2z}{3}$, and $\frac{0x+2y+1z}{3}$.
\end{definition}

Bob's map can be defined using the three numbers $P_{\overline{Y}|Y}(\overline{0},1)=r$, $P_{\overline{Y}|Y}(\overline{0},2)=s$, and $P_{\overline{Y}|Y}(\overline{0},3)=t$. Since Bob's map is a binarization, we have that $P_{\overline{Y}|Y}(\overline{1},1)=1-r$, $P_{\overline{Y}|Y}(\overline{1},2) = 1-s$, and $P_{\overline{Y}|Y}(\overline{1},3) = 1-t$. The probability distribution $X\overline{Y}Z$ is as follows, using the same notation as before:

\begin{center}
\begin{tabular}{|C{0.9cm}||c|c|c|}
    \hline 
    ~ $X$  & 1 & 2 & 3  \\ 
    $\overline{Y} (Z)$ &&& \\
    \hhline{|=#=|=|=|}
     & $2r$ (0) & $2s$ (0) & $2t$ (0) \\
    {\centering $\overline{0}$} & $s$ (3) & $4r$ (1) & $r$ (2) \\
    & $4t$ (5) & $t$ (6) & $4s$ (4) \\ \hline
     & $2-2r$ (0) & $2-2s$ (0) & $2-2t$ (0) \\
    $\overline{1}$ & $1-s$ (3) & $4-4r$ (1) & $1-r$ (2) \\
    & $4-4t$ (5) & $1-t$ (6) & $4-4s$ (4) \\ \hline
\end{tabular}
\end{center}

As mentioned in \cite{GiReWo02}, if Eve receives $z \neq 0$, she knows what $X$ is, meaning that $X|Z=z$ is constant and $X \perp\!\!\!\perp \overline{Y} | Z=z$. Therefore, we focus our attention on the case that $Z=0$. We consider the same map $P_{\overline{Z}|Z}$ as mentioned in the proof of Proposition 4 of \cite{GiReWo02}. In this map, the nonzero values for $Z$, namely $1$, $2$, $3$, $4$, $5$, $6$ are mapped to $\overline{0}$ with probabilities $c$, $e$, $a$, $f$, $b$, and $d$ respectively, and they are mapped to $\overline{1}, \dots, \overline{6}$ with probabilities $1-c$, $1-e$, $1-a$, $1-f$, $1-b$, $1-d$ respectively. The value $Z=0$ is mapped to $\overline{0}$ with probability $1$. Under this map, the probability distribution $P_{X\overline{Y}|\overline{Z}=\overline{0}}$ is as follows:

\begin{center}
\begin{tabular}{|C{1.5cm}||c|c|c|}
    \hline 
    ~ $X$  & 1 & 2 & 3  \\ 
    $\overline{Y}$ &&& \\
    \hhline{|=#=|=|=|}
    \\[-1em]
    $\overline{0}$ & $2r+as+4bt$ & $4cr+2s+dt$ & $er+4fs+2t$ \\
    \hline
    \\[-1em]
    $\overline{1}$ & $(2+a+4b)-(2r+as+4bt)$ & $(4c+2+d)-(4cr+2s+dt)$ & $(e+4f+2)-(er+4fs+2t)$ \\
    \hline
\end{tabular}
\end{center}

In order to have $(X \perp\!\!\!\perp \overline{Y}) | \overline{Z}=\overline{0}$, we must have the following:
\[\frac{2r+as+4bt}{2+a+4b}=\frac{4cr+2s+dt}{4c+2+d}=\frac{er+4fs+2t}{e+4f+2}.\]
In the case that the denominators of these fractions are 0 and the numerators nonzero, the independence condition is satisfied. If both the numerator and denominator of a fraction are $0$, it can be ignored since it imposes no additional conditions.
In \cite{GiReWo02}, it is proven that for any $r$, $s$, $t$ there exist $a, b, c, d, e, f$ satisfying the above equations using a topological argument, but here we demonstrate this in a constructive manner using the ITV:

\begin{theorem}
For numbers $r,s,t\in \mathbb{R}$, there exist $a,b,c,d,e,f\in [0,1]$ such that 
\[\frac{2r+as+4bt}{2+a+4b}=\frac{4cr+2s+dt}{4c+2+d}=\frac{er+4fs+2t}{e+4f+2}=\tau(r,s,t).\]
\end{theorem}

\begin{proof}
Note that $\tau(r,s,t)=\tau(s,t,r)=\tau(t,r,s)$. This means that if we can find satisfactory $a,b\in [0,1]$ such that 
\[\frac{2r+as+4bt}{2+a+4b}=\tau(r,s,t)\]
for any $r,s,t\in \mathbb{R}$, then by symmetry we can find satisfactory $c,d\in [0,1]$ such that
\[\frac{2s+dt+4cr}{2+d+4c}=\tau(s,t,r)=\tau(r,s,t)\]
for the same set of $r,s,t\in [0,1]$. Similarly, we can also find $e,f\in [0,1]$ such that
\[\frac{2t+er+4fs}{2+e+4f}=\tau(t,r,s)=\tau(r,s,t)\]
for the same set of $r$, $s$, and $t$. This means that we only need to show that for all $r,s,t\in\mathbb{R}$, there exist $a,b\in [0,1]$ such that 

\begin{equation}\label{eq:1}
\frac{2r+as+4bt}{2+a+4b}=\tau(r,s,t).
\end{equation}

If $r=s=t$, then both sides of the equation above are equal to $r$ regardless of the choice of $a,b$. Now, assume that not all three of $r,s,t$ are equal. Note that since the left and right sides of the equation above are computed from weighted averages of $r$, $s$, and $t$, we can scale the variables by a nonzero constant or add a real number without changing the equation. We can subtract $\min(r,s,t)$ from all of variables and since all of the variables are not equal, we can divide these new variables by $\max(r,s,t)-\min(r,s,t)\neq 0$. We have now transformed $(r,s,t)$ to  $(r',s',t')$, a permutation of $(0,x,1)$, where $x\in [0,1]$.

Our new variables are thus one of the cyclic shifts of $(0,x,1)$ or $(0,1,x)$ for some $x\in [0,1]$. In the latter case, we can multiply the triple by $-1$ and add 1 to get a cyclic shift of $(0,1-x,1)$. This means that we can assume without loss of generality that $(r',s',t')$ is some cyclic shift of $(0,x,1)$ for some $x\in [0,1]$. 

Since the ITV is invariant under cyclic shifts, we now have $\tau(r',s',t')=\tau(0,x,1)$, which is the median of $\frac{x}{3}$, $\frac{2}{3}$, and $\frac{1+2x}{3}$. Note that $\frac{x}{3}$ is the least of these three, so the median is $\frac{\min(2,1+2x)}{3}$. If $x<\frac{1}{2}$, then this is $\frac{1+2x}{3}$, and if $x\geq \frac{1}{2}$, then this is $\frac{2}{3}$.

We now explicitly write out satisfactory $a$ and $b$ which satisfy equation \ref{eq:1} based on which of these cyclic shifts of $(0,x,1)$ that $(r,s,t)$ has been transformed into:
\begin{itemize}
    \item $(r',s',t')=(0,x,1)$: If $x<\frac{1}{2}$, then we can take $a=0$ and $b=\frac{1+2x}{4-4x}$. If $x\geq \frac{1}{2}$, then we can take $a=0$ and $b=1$.
    \item $(r',s',t')=(1,0,x)$: If $x<\frac{1}{2}$, then we can take $a=0$ and $b=1$. If $x\geq \frac{1}{2}$, then we can take $a=1$ and $b=0$.
    \item $(r',s',t')=(x,1,0)$: If $x<\frac{1}{2}$, then we can take $a=1$ and $b=0$. If $x\geq \frac{1}{2}$, then we can take $a=1$ and $b=\frac{3}{8}(2x-1)$.
\end{itemize}
This covers all of the cases, so we are done.
\end{proof}

This resolves the $N=1$ case. We attempt to extend the use of the ITV for $N=2$. Bob's map may be parameterized by the values $a_{ij}:=P_{\overline{Y^2}|Y^2}(0,ij)$ for $i,j\in\{1,2,3\}$, where $Y^2=ij$ represents its two components, $(Y_1,Y_2)=(i,j)$. In this case, we cannot focus on the $Z^2=00$ case alone, because if $Z^2=01$ for example, Eve is unsure of whether Alice has $X^2=12$, $X^2=22$, or $X^2=32$. 
If neither component of $Z^2$ are $0$, Eve will be certain of what Alice has, so $X$ and $\overline{Y}$ are already independent in these cases and no processing is necessary.
This means that we must repeat the above procedure for $\overline{Z^2}=\overline{00},\overline{01},\overline{02},\dots ,\overline{06}, \overline{10},\overline{20},\dots, \overline{60}$. Given a $\overline{Z^2}$ value of $\overline{ij}$, the fractions will be of the form \[\frac{P\left(\overline{X^2}=0, Y^2=11, \overline{Z^2}=\overline{ij}\right)}{P\left(Y^2=11, \overline{Z^2}=\overline{ij}\right)} = \dots = \frac{P\left(\overline{X^2}=0, Y^2=11, \overline{Z^2}=\overline{ij}\right)}{P\left(Y^2=33, \overline{Z^2}=\overline{ij}\right)}.\] 

One natural possibility that is extendable to $N\geq 3$ for the target value for the fractions corresponding to these $z^2$ values are the following:

\begin{conjecture}\label{n=2for214}
Define $a_{ij}:=P_{\overline{Y^2}|Y^2}(0,rs)$ for $r,s\in\{1,2,3\}$, and let the target values $\tau_2:\{00,01,02,\dots, $ $06, 10,20,\dots, 60\}\to [0,1]$ be defined as follows:
\begin{itemize}
    \item $\tau_2(0i) = \tau(a_{1j},a_{2j},a_{3j})$, where $j=\lceil \frac{i}{2}\rceil$ and $1\leq i\leq 6$,
    \item $\tau_2(i0) = \tau(a_{j1},a_{j2},a_{j3})$, where $j=\lceil \frac{i}{2}\rceil$ and $1\leq i\leq 6$,
    \item $\tau_2(00)=\tau(\tau(a_{11},a_{12},a_{13}),\tau(a_{21},a_{22},a_{23}),\tau(a_{31},a_{32},a_{33}))$.
\end{itemize}
Then, there exists a channel $P_{\overline{Z^2}|Z^2}$ such that $P\left(\overline{Y^2}=0|X^2, \overline{Z^2}=\overline{z}\right)=\tau_2(z)$ for all $z$ in the domain of $\tau_2$.
\end{conjecture}

The choice of $j=\lceil \frac{i}{2}\rceil$ is motivated by the fact that in this distribution, if Eve receives $Z=i$, then she knows that Bob has $Y=\lceil \frac{i}{2}\rceil=j$. We observe that these target values give the correct values for a particular class of Bob's strategies which we term product strategies:

\begin{definition}
A \textit{product strategy} for Bob (who has $Y^N=Y_1Y_2\dots Y_N$) is a binarization of $Y^N$ so that the $N$-dimensional matrix of probabilities $P\left(\overline{Y^N} = 0|Y^N = y\right)$ for $y \in \mathcal{Y}^N$ is the tensor product of the $N$ vectors $P\left(\overline{Y^N} = 0|Y_1 = y_1\right)$, \ldots, $P\left(\overline{Y^N} = 0|Y_N = y_N\right)$ where each $y_i$ takes on every value in $\mathcal{Y}$.
\end{definition}

The reason that this target value choice works for product strategies is due to the following property of the ITV:

\begin{theorem}\label{ITVprop}
For real numbers $b_0,b_1,b_2,c_0,c_1,c_2\in [0,1]$,
\[\tau(\tau(b_0c_0,b_1c_0,b_2c_0),\tau(b_0c_1,b_1c_1,b_2c_1),\tau(b_0c_2,b_1c_2,b_2c_2))=\tau(b_0,b_1,b_2)\tau(c_0,c_1,c_2).\]
\end{theorem}

\begin{proof}
Note that $\tau(b_0c_0,b_1c_0,b_2c_0)=c_0\tau(b_0,b_1,b_2)$ because we can factor the $c_0$ from the set of weighted averages considered when calculating the ITV. We can use this repeatedly to get that the left-hand side is equal to
\[\tau(c_0\tau(b_0,b_1,b_2),c_1\tau(b_0,b_1,b_2),c_2\tau(b_0,b_1,b_2))=\tau(b_0,b_1,b_2)\tau(c_0,c_1,c_2).\]
\end{proof}

Note that if the vectors $P\left(\overline{Y^2} = 0|Y_1 = 0\right)$ and $P\left(\overline{Y^2} = 0|Y_2 = 0\right)$ are $(b_0,b_1,b_2)$ and $(c_0,c_1,c_2)$, respectively, then the product strategy this corresponds to is the map $a_{ij}=b_ic_j$ for all $i,j\in \{0,1,2\}$. Using theorem \ref{ITVprop}, the target value for $Z^2=00$ is equal to $\tau(b_0,b_1,b_2)\tau(c_0,c_1,c_2)$. Therefore, we can use our results from the $N=1$ case for each of components for $Y^2$, and then multiply them together to construct our map.

However, when Bob does not use product strategies, there exist strategies for Bob for which Eve cannot set each of the fractions to their desired target values:

\begin{theorem}\label{ITVcounter}
    For $N=2$, if Bob's transition probabilities are $(a_{ij})=\begin{psmallmatrix}1 & 1&0\\0 &0&0\\0&0& 1\end{psmallmatrix}$ with an ITV of $\tau_2$ defined in conjecture \ref{n=2for214}, then there does not exist a channel $P_{\overline{Z^2}|Z^2}$ such that $P\left(\overline{Y^2}=0|X^2, \overline{Z^2}=\overline{z}\right)=\tau_2(z)$ for all $z$ in the domain of $\tau_2$.
\end{theorem}

\begin{proof}
    We create a function in \textit{Mathematica} which given Bob's transition probabilities and all of the ITV either outputs a channel $P_{\overline{Z^2}|Z^2}$ that Eve can use or states that no possible channel exists. 
    This is done by taking the equations corresponding to independence between $\overline{X^2}$ and $Y^2$ given $\overline{Z^2}$ and manipulating them, reducing the problem into linear equations in the probabilities that comprise $P_{\overline{Z^2}|Z^2}$, and solving these equations with the \texttt{LinearProgramming} command.

    Through this function, for this set of Bob's transition probabilities and ITV it was found that no satisfactory $P_{\overline{Z^2}|Z^2}$ existed, as desired.
\end{proof}

In particular, by testing distributions of random numbers, we have found that such target values for Eve exist approximately 80\% of the time. Due to this high percentage and generalizability, we hope to use  similarly constructed target values (e.g. $\tau_3(00i)=\tau(\tau(a_{11i},a_{12i},a_{13i}),\tau(a_{21i},a_{22i},a_{23i}),\tau(a_{31i},a_{32i},a_{33i}))$ and so on) for $N=2$ and beyond. 

We now address the question of what the shape of Eve's channel should be in the case of $N=2$. Recall that in the case of $N=1$, for any binarization that Bob chooses it is possible for Eve to make $X \perp\!\!\!\perp \overline{Y}|\overline{Z}$ using a channel in which $Z=0$ is sent to $\overline{Z}=\overline{0}$ with probability $1$, and no nonzero $Z=i$ is sent to $\overline{Z}=\overline{j}$ with positive probability for any $j \neq 0,i$. Since $X \perp\!\!\!\perp \overline{Y}|Z=i$ for $i\in \{1, 2, \dots, 6 \}$, we can freely change the values of the transition probabilities from $Z=i$ to $\overline{Z}=\overline{0}$ in order to make $X \perp\!\!\!\perp \overline{Y}|\overline{Z}=\overline{0}$. We consider a natural extension of this idea to the case $N=2$, namely channels $P_{\overline{Z^2}|Z^2}$ such that the only nonzero transition values are the following:
\begin{itemize}
    \item $P_{\overline{Z^2}|Z^2}(\overline{ij},ij)$ for all $i$ and $j$,
    \item $P_{\overline{Z^2}|Z^2}(\overline{i0},ij)$ for all $i$ and $j$,
    \item $P_{\overline{Z^2}|Z^2}(\overline{0j},ij)$ for all $i$ and $j$,
    \item $P_{\overline{Z^2}|Z^2}(\overline{00},00) = 1$.
\end{itemize}
In other words, if one writes out the $Z^2$ values in a 2D grid, with corners $00$, $06$, $60$, and $66$, the only transitions allowed are within a column or row, or from an element to itself. We prove that this cannot work.

\begin{theorem}
There exists a binarization $P_{\overline{Y^2}|Y^2}$ such that no channel satisfying the properties above results in $X \perp\!\!\!\perp \overline{Y}|\overline{Z}=\overline{00}$.
\end{theorem}
\begin{proof}
Consider the binarization $P_{\overline{Y^2}|Y^2}$ defined by $P_{\overline{Y^2}|Y^2}(0, i) = 1$ for $i=11, 33$ and $P_{\overline{Y^2}|Y^2}(0, i) = 0$ otherwise. Observe that for the map above, the only transitions to $Z=\overline{00}$ are from $n0$ or $0n$ for $n = 1, 2, \dots, 6$, or from $00$. We know that the transition probability from $00$ to $\overline{00}$ is $1$. Let the transition probabilities $P_{\overline{Y^2}|Y^2}(\overline{n0},00)$ for $n = 1, 2, \dots, 6$ be $a_1, \dots, f_1$ respectively, and similarly let the transition probabilities $P_{\overline{Y^2}|Y^2}(\overline{0n},00)$ be $a_2, \dots, f_2$ respectively. Then the equations \[\frac{P\left(\overline{X^2}=\overline{0}, Y^2=11, \overline{Z^2}=\overline{ij}\right)}{P\left(Y^2=11, \overline{Z^2}=\overline{00}\right)} = \dots = \frac{P\left(\overline{X^2}=\overline{0}, Y^2=11, \overline{Z^2}=\overline{00}\right)}{P\left(Y^2=33, \overline{Z^2}=\overline{ij}\right)}\] become
\begin{align*}
    &\frac{2}{2+a_1+a_2+4b_1+4b_2} = \frac{4c_2}{2+a_1+4b_1+4c_2+d_2} = \frac{4b_1+e_2}{2+a_1+4b_1+e_2+4f_2} = \frac{4c_1}{2+a_2+4b_2+4c_1+d_1} \\
    &=0 \\
    &=\frac{d_1}{2+4c_1+d_1+e_2+4f_2} = \frac{4b_2+e_1}{2+a_2+4b_2+e_1+4f_1} = \frac{d_2}{2+4c_2+d_2+e_1+4f_1} = \frac{2}{2+e_1+e_2+4f_1+4f_2}
\end{align*}
The first fraction can never be equal to 0 (the fifth expression in the equality) for any values of $a_1, \dots, f_1$ and $a_2, \dots, f_2$, so we are done.
\end{proof}

By providing an explicit construction for the 1D case using the ITV, the results in this section suggest a possible approach for generalizing the existence proof first given in \cite{GiReWo02}, which promises to generalize to $X^NY^NZ^N$. However, we illustrate some difficulties in performing performing a straightforward generalization to higher dimensions, in particular the case where one of the parties, say Bob, does not use a product strategy.

\section{A Family of Candidate Distributions}
\label{family}

We now examine another distribution given in \cite{RenWol03} which is believed to be bound secret. The unnormalized probability table is shown below.
\begin{center}
\begin{tabular}{|l||c|c|c|c|}
    \hline 
    ~ $X$  & 0 & 1 & 2 & 3  \\ 
    $Y$ &&&& \\
    \hhline{|=#=|=|=|=|}
    0 & 1/8 & 1/8 & $a$ & $a$ \\ \hline
    1 & 1/8 & 1/8 & $a$ & $a$ \\ \hline
    2 & $a$ & $a$ & $1/4$ & 0 \\ \hline
    3 & $a$ & $a$ & 0 & $1/4$ \\
    \hline
\end{tabular}
\begin{align*}
    Z &\equiv X + Y \pmod{2} \text{ if } X, Y \in \{0, 1 \}, \\
    Z &\equiv X \pmod{2} \text{ if } X, Y \in \{2, 3 \}, \\
    Z &= (X, Y) \text{ otherwise.}
\end{align*}
\end{center}
It has already been shown that for any $a > 0$, we have $I(X:Y\downarrow Z)>0$ \cite{RenWol03}, so to prove that this distribution is bound secret, we only need to show that $S(X:Y||Z)=0$ for any fixed $a$. As with the previous distribution, we conjecture the following:

\begin{conjecture}
\label{4x4conj}
There exists a value $a$ such that for the distribution $XYZ$ above, for any $N \geq 1$ we have
\[ \forall \overline{X^N}, \overline{Y^N}, ~ \exists \overline{Z^N} \text{ such that } (\overline{X^N} \perp\!\!\!\perp \overline{Y^N}) | \overline{Z^N} \]
where the channels processing $X^N$ and $Y^N$ are assumed to be binarizations.
\end{conjecture}

We have not been able to establish the above statement for the case $N=1$, but have made some progress in ruling out possible simplifications. For example, in the case of the previous distribution it was possible to prove the above statement for $N=1$ while only allowing binarizations of $Y$, which is a strictly harder task than if we had allowed binarizations of $X$ and $Y$. We show that the analogous way of strengthening the above claim in the case of the current distribution for $N=1$ cannot work.

In what follows, we consider channels $P_{\overline{Z}|Z}$ as follows. Observe that $X\perp\!\!\!\perp Y |Z=z$ (and therefore $\overline{X}\perp\!\!\!\perp \overline{Y} |Z=z$ for all $\overline{X}$, $\overline{Y}$) for all $z \not\in \{0, 1\}$, so there is no reason to have a transition to $\overline{Z}=\overline{z}$ except from $Z=z$ itself. We allow all other transitions except those from $Z=0$ or $Z=1$ to $\overline{Z}=\overline{z}$ with $\overline{z} \neq \overline{0}, \overline{1}$. These transitions would be counterproductive for Eve because it is already the case that $\overline{X}\perp\!\!\!\perp \overline{Y} |Z=z$, and after this transitions it might be the case that $\overline{X}$ and $\overline{Y}$ are dependent given $\overline{Z}=\overline{z}$. 

With these restrictions, we parameterize the possible $\overline{Z}$ channels as follows. Let $P_{\overline{Z}|Z}(\overline{1}, 0) = \alpha$ and $P_{\overline{Z}|Z}(\overline{0}, 1) = \beta$, so that $P_{\overline{Z}|Z}(\overline{1}, 1) = 1-\beta$ and $P_{\overline{Z}|Z}(\overline{0}, 0) = 1-\alpha$. Let the transition probabilities from $Z=(0,2)$, $(1,2)$, $(0,3)$, $(1,3)$, $(2,0)$, $(3,0)$, $(2,1)$, and $(3,1)$ to $\overline{Z}=\overline{0}$ be $a_0$, $b_0$, $c_0$, $d_0$, $e_0$, $f_0$, $g_0$, and $h_0$ respectively. Similarly define $a_1, \dots, h_1$ to be the transition probabilities from $Z$-values not equal to $0$ or $1$ to $\overline{Z} = \overline{1}$. This is illustrated in Tables \ref{Fig1} and \ref{Fig2}, where the number in each cell is the transition probability to $\overline{Z}=\overline{0}$ (Table \ref{Fig1}) or $\overline{1}$ (Table \ref{Fig2}) from the $Z$ value corresponding to the $XY$ value for that cell.

\begin{center}
\begin{table}\begin{tabular}{|l||c|c|c|c|}
    \hline 
    ~ $X$  & 0 & 1 & 2 & 3  \\ 
    $Y$ &&&& \\
    \hhline{|=#=|=|=|=|}
    0 & $1- \alpha$ & $\beta$ & $e_0$ & $f_0$ \\ \hline
    1 & $\beta$ & $1-\alpha$ & $g_0$ & $h_0$ \\ \hline
    2 & $a_0$ & $b_0$ & $1-\alpha$ & 0 \\ \hline
    3 & $c_0$ & $d_0$ & 0 & $\beta$ \\
    \hline
\end{tabular}\caption{\label{Fig1}Transition probabilities to $\overline{Z}=\overline{0}$.}\end{table}~~~\begin{table}\begin{tabular}{|l||c|c|c|c|}
    \hline 
    ~ $X$  & 0 & 1 & 2 & 3  \\ 
    $Y$ &&&& \\
    \hhline{|=#=|=|=|=|}
    0 & $\alpha$ & $1-\beta$ & $e_1$ & $f_1$ \\ \hline
    1 & $1-\beta$ & $\alpha$ & $g_1$ & $h_1$ \\ \hline
    2 & $a_1$ & $b_1$ & $1-\alpha$ & 0 \\ \hline
    3 & $c_1$ & $d_1$ & 0 & $\beta$ \\
    \hline
\end{tabular}\caption{\label{Fig2}Transition probabilities to $\overline{Z}=\overline{1}$.}\end{table}

\end{center}

\begin{theorem}
\label{Ybarneeded}
For all possible values of $a>0$ in the distribution given in \cite{RenWol03}, there exists a binarization $P_{\overline{Y}|Y}$ such that for all $\overline{Z}$ in the form given by Tables \ref{Fig1} and \ref{Fig2}, $X\not\perp\!\!\!\perp\overline{Y}|\overline{Z}$.
\end{theorem}
\begin{proof}
Fix $a > 0$. In order for $X\perp\!\!\!\perp\overline{Y}|\overline{Z}$, we need  \[ \frac{P\left(X=0, \overline{Y}=0, \overline{Z}=0\right)}{P\left(X=0, \overline{Z}=0\right)} = \dots = \frac{P\left(X=3, \overline{Y}=0, \overline{Z}=0\right)}{P\left(X=3, \overline{Z}=0\right)} \] and \[ \frac{P\left(X=0, \overline{Y}=0, \overline{Z}=1\right)}{P\left(X=0, \overline{Z}=1\right)} = \dots = \frac{P\left(X=3, \overline{Y}=0, \overline{Z}=1\right)}{P\left(X=3, \overline{Z}=1\right)}. \] We claim that there exists a binarization $P_{\overline{Y}|Y}$ such that for all $P_{\overline{Z}|Z}$ in the form given by tables \ref{Fig1} and \ref{Fig2}, the two equations $\frac{P\left(X=2, \overline{Y}=0, \overline{Z}=0\right)}{P\left(X=2, \overline{Z}=0\right)} = \frac{P\left(X=3, \overline{Y}=0, \overline{Z}=0\right)}{P\left(X=3, \overline{Z}=0\right)}$ and $\frac{P\left(X=2, \overline{Y}=0, \overline{Z}=1\right)}{P\left(X=2, \overline{Z}=1\right)} = \frac{P\left(X=3, \overline{Y}=0, \overline{Z}=1\right)}{P\left(X=3, \overline{Z}=1\right)}$ cannot be satisfied. Let this binarization be defined by $P_{\overline{Y}|Y}(\overline{0}, 0) = w$, $P_{\overline{Y}|Y}(\overline{0}, 1) = x$, $P_{\overline{Y}|Y}(\overline{0}, 2) = y$, and $P_{\overline{Y}|Y}(\overline{0}, 3) = z$ (we will specify $w, x, y, z$ later). Expanding, the two equations are
\begin{align*}
    \frac{a e_0 w + a g_0 x + 2 y(1-\alpha)}{a e_0 + ag_0 + 2(1-\alpha)} &= \frac{af_0 w + ah_0 x + 2z \beta}{af_0 + ah_0 + 2\beta} \\
    \frac{a e_1 w + a g_1 x + 2y\alpha}{ae_1 + ag_1 + 2\alpha} &= \frac{af_1 w + ah_1 x + 2z(1-\beta)}{af_1 + ah_1 + 2(1-\beta)}
\end{align*}

Observe that, since both of these equations are weighted averages of $w$, $x$, $y$, and $z$, the equations do not change if an affine transformation is applied to $w$, $x$, $y$, and $z$ simultaneously. So, in the case that $w \neq x$ (which will hold for the particular binarization that makes $X$ and $\overline{Y}$ dependent regardless of $P_{\overline{Z}|Z}$), we can let $w=0$ and $x=1$, WLOG. After a bit of simplification, the equations become
\begin{align*}
    \frac{g_0 + 2 \left(\frac{y}{a}\right) (1-\alpha)}{e_0 + g_0 + \frac{2(1-\alpha)}{a}} &= \frac{h_0 + 2 \left(\frac{z}{a}\right) \beta}{f_0 + h_0 + \frac{2\beta}{a}} \\
    \frac{g_1 + 2 \left(\frac{y}{a}\right) \alpha}{e_1 + g_1 + \frac{2\alpha}{a}} &= \frac{h_1 + 2 \left(\frac{z}{a}\right) (1-\beta)}{f_1 + h_1 + \frac{2(1-\beta)}{a}}
\end{align*}
We claim that the values $y=2$ and $z=3a+3$ work (i.e. there are no possible values $0\leq e_0, \dots, h_1, \alpha, \beta\leq 1$ that satisfy the equations). \newline

To prove this, we can again view each fraction as a weighted average. For example, the first fraction in the first equation is a weighted average of the values $0$, $1$, and $y$ with weights $e_0$, $g_0$, and $\frac{2(1-\alpha)}{a}$ respectively. Since $y, z > 1$, the minimum possible value of this fraction is achieved when $e_0 = g_0 = 1$ and the maximum is achieved when $e_0 = g_0 = 0$. Similarly for the second fraction in the first equation, the minimum is achieved when $f_0 = h_0 = 1$ and the maximum is achieved when $f_0 = h_0 = 0$. Since the maximum of the second fraction is $z$, which is always larger than $y$ for the proposed values $y=2$ and $z=3a+3$, the first equation is solvable if and only if the minimum of the second fraction is less than or equal to the maximum of the first fraction, that is \[ y \geq \frac{1+2 \left(\frac{z}{a} \right) \beta}{2 + \left(\frac{2}{a} \right) \beta}.\] Solving for $\beta$, we have \[ \beta \leq \frac{a(2y-1)}{2z-2y}.\] Applying the same analysis to the second equation gives that \[ 1-\beta \leq \frac{a(2y-1)}{2z-2y}\] since the second equation is exactly the same as the first, except with different parameters $e_1, \dots, h_1$ and with $1-\alpha$ substituted for $\alpha$, and $\beta$ substituted for $1-\beta$. We now claim that \[\frac{a(2y-1)}{2z-2y}<\frac{1}{2}\] for our proposed values $y=2$ and $z=3a+3$, and therefore satisfying both of the above inequalities is impossible. We have that $\frac{a(2y-1)}{2z-2y} = \frac{3a}{2(3a+1)}<\frac{3a}{2(3a)} = \frac{1}{2}$ as desired.
\end{proof}

Theorem \ref{Ybarneeded} suggests that establishing Conjecture \ref{4x4conj} would be significantly more involved than proving Conjecture \ref{strong326}, because we would need to consider $8$ different variables in creating the channel $P_{\overline{Z}|Z}$. We conjecture one possible minor simplification regarding $P_{\overline{Z}|Z}$.

\begin{definition}[\cite{mackay2003information}]
Consider a binary random variable $B$. Call a channel $P_{\overline{B}|B}$ a \textit{Z-shaped channel} if at least one of $P_{\overline{B}|B}(\overline{0},0)$, $P_{\overline{B}|B}(\overline{1},0)$, $P_{\overline{B}|B}(\overline{0},1)$, and $P_{\overline{B}|B}(\overline{1},1)$ is zero.
\end{definition}

\begin{center}
\begin{figure}
\begin{tikzpicture}

\coordinate (A) at (0,0);
\coordinate (B) at (2,0);
\coordinate (C) at (2,2);
\coordinate (D) at (0,2);
\coordinate (E) at (-0.5,1);
\coordinate (F) at (2.5,1);

\draw [->] (A) -- (B);
\draw [->] (A) -- (2,1.9);
\draw [->] (D) -- (1.9,2);

\node[left] at (A) {$1$};
\node[right] at (B) {$\overline{1}$};
\node[right] at (C) {$\overline{0}$};
\node[left] at (D) {$0$};
\node[left] at (E) {$B$};
\node[right] at (F) {$\overline{B}$};

\end{tikzpicture}
\caption{Example of Z-shaped channel, with $P_{\overline{B}|B}(\overline{1},0)=0$.}
\end{figure}
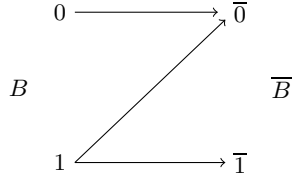
\end{center}

We observe the following:
\begin{theorem}
Let $B$ be a binary random variable. Any channel $P_{\overline{B}|B}$ is equivalent to using a $Z$-shaped channel with some probability, and performing a fair coin flip (with the outcome of the coin determining the outputted bit) otherwise. 
\end{theorem}
\begin{proof}
Let $P_{\overline{B}|B}(\overline{0},0) = a$, $P_{\overline{B}|B}(\overline{1},0) = c$, $P_{\overline{B}|B}(\overline{0},1) = b$, and $P_{\overline{B}|B}(\overline{1},1) = d$. Suppose that $\min(a,b,c,d) = b$ (the other cases are similar). Then it can be verified that $P_{\overline{B}|B}$ is equivalent to performing a fair coin flip with probability $2b$, and using the channel $P_{\overline{B}'|B}$ defined by $P_{\overline{B}'|B}(\overline{0}, 0) = 1$, $P_{\overline{B}'|B}(\overline{0}, 1) = \frac{c-b}{a-b}$, and $P_{\overline{B}'|B}(\overline{1},1) = \frac{d-b}{a-b}$ otherwise. If $a=b=1-c=1-d\neq \frac{1}2$, then $P_{\overline{B}|B}$ is a biased coin, which is a weighted average of a fair coin and a Z-shaped channel always outputting $\overline{0}$, if $a>\frac{1}2$, or a Z-shaped channel always outputting $\overline{1}$, if $a<\frac{1}2$. Finally, if $a=b=c=d=\frac{1}{2}$, then $P_{\overline{B}|B}$ is already a fair coin, so we are done.
\end{proof}

In order to make $\overline{X}\perp\!\!\!\perp \overline{Y}|\overline{Z}$, it seems counterproductive for Eve to throw away her information (by using a coin flip) with some probability. For this reason, we conjecture that we can restrict the space of channels $P_{\overline{Z}|Z}$ to those such that the transitions between $0$ and $1$ and $\overline{0}$ and $\overline{1}$ form a Z-shaped channel.

\begin{conjecture}
There exists an $a>0$ in the above distribution such that for all binarizations $P_{\overline{X}|X}$ and $P_{\overline{Y}|Y}$, there exists $P_{\overline{Z}|Z}$ such that at least one of $P_{\overline{Z}|Z}(\overline{0}, 0)$, $P_{\overline{Z}|Z}(\overline{1}, 0)$, $P_{\overline{Z}|Z}(\overline{0}, 1)$, and $P_{\overline{Z}|Z}(\overline{1}, 1)$ is $0$ and
\[  (\overline{X} \perp\!\!\!\perp \overline{Y}) | \overline{Z}. \]
\end{conjecture}

Another natural approach for the $N=1$ case of this distribution is to try and find a suitable ITV, similar to how the $N=1$ case for the distribution introduced in Section 6.
In particular, one would need to find a function $\upsilon : \mathbb{R}^4\to \mathbb{R}$ which maps the four values corresponding to Bob's transition map ($P_{\overline{Y}|Y}(\overline{0},r-1)$ for $r\in \{1,2,3,4\}$) to the target value for each of the four fractions
\[\frac{P\left(X=0, \overline{Y}=0, \overline{Z}=\overline{0}\right)}{P\left(X=0, \overline{Z}=\overline{0}\right)} = \dots = \frac{P\left(X=3, \overline{Y}=0, \overline{Z}=\overline{0}\right)}{P\left(X=3, \overline{Z}=\overline{0}\right)},\]
which correspond to $X\perp\!\!\!\perp\overline{Y}|\overline{Z}=\overline{0}$.
Furthermore, in order for this approach to generalize for $N\geq 2$, an ITV function $\upsilon_2 : \mathbb{R}^{16}\to \mathbb{R}$ must be chosen similar to conjecture \ref{n=2for214}.
In particular, given the transition values $b_{rs}:=P_{\overline{Y^2}|Y^2}(0,(r-1)(s-1))$ for $r,s\in\{1,2,3,4\}$ (where $(r-1)(s-1)$ indicates the concatenation of $Y_1=r-1$ and $Y_2=s-1$), one possible candidate for $\upsilon_2$ is as follows:
\[\upsilon_2(b_{11},b_{12},b_{13},\dots ,b_{44}):=\upsilon(\upsilon(b_{11},b_{12},b_{13},b_{14}),\upsilon(b_{21},b_{22},b_{23},b_{24}),\dots ,\upsilon(b_{41},b_{42},b_{43},b_{44})).\]
One drawback of such an ITV for the $N=2$ case is that the two components of $Y^2:=Y_1Y_2$ are being treated differently: if the $b_{ij}$ are placed in a 4 by 4 table, the $\upsilon$ is taken over the rows first rather than the columns first, giving priority to $Y_1$. However, for a special class of $\upsilon$ ITV, this issue is not present:

\begin{definition}\label{rowcol}
    Call an ITV $\upsilon : \mathbb{R}^4\to \mathbb{R}$ \textit{row-column equivalent} if the following statement is true for all $b_{ij}\in \mathbb{R}$ ($i,j\in\{1,2,3,4\}$):
    \[\upsilon(\upsilon(b_{11},b_{12},b_{13},b_{14}),\upsilon(b_{21},b_{22},b_{23},b_{24}),\dots ,\upsilon(b_{41},b_{42},b_{43},b_{44}))=\]
    \[\upsilon(\upsilon(b_{11},b_{21},b_{31},b_{41}),\upsilon(b_{12},b_{22},b_{32},b_{42}),\dots ,\upsilon(b_{14},b_{24},b_{34},b_{44})).\]
\end{definition}

\begin{theorem}\label{weightavg}
    All ITV of the form $\upsilon(r,s,t,u):=w_1r+w_2s+w_3t+w_4u$ with $w_1,w_2,w_3,w_4\in \mathbb{R}$ are row-column equivalent.
\end{theorem}
\begin{proof}
    Note that
    \[\upsilon(\upsilon(b_{11},b_{12},b_{13},b_{14}),\upsilon(b_{21},b_{22},b_{23},b_{24}),\dots ,\upsilon(b_{41},b_{42},b_{43},b_{44}))=\upsilon\left(\sum_{i=1}^4 w_ib_{1i},\sum_{i=1}^4 w_ib_{2i},\dots,\sum_{i=1}^4 w_ib_{4i}\right)\]
    \[=\sum_{j=1}^4\sum_{i=1}^4w_jw_ib_{ij}.\]
    Since this quantity is symmetric in $i,j$, the equation in the definition of row-column equivalent is satisfied.
\end{proof}

We now consider ITV of the form $\upsilon(r,s,t,u):=w_1r+w_2s+w_3t+w_4u$ due to the above property.
Note that the fractions of the form $\frac{P\left(X=0, \overline{Y}=0, \overline{Z}=\overline{0}\right)}{P\left(X=0, \overline{Z}=\overline{0}\right)}$ are weighted averages of the values corresponding to Bob's transition map. As a result, the ITV should also be a weighted average of these values. This means that the condition $w_1+w_2+w_3+w_4=1$ should be imposed on this class of ITV.

Now, we consider the $N=2$ case, with the $\upsilon$ and $\upsilon_2$ defined above. We can classify Bob's transition values $\{b_{ij}\}$ based on their $\upsilon_2$ value. In particular, since $\upsilon$ is a weighted average, we can explicitly construct transformations that preserve the $\upsilon_2$ value:

\begin{definition}
    For a set of Bob transition values $\{b_{ij}\}$, a \textit{row transformation} is defined as considering an $i\in \{1,2,3,4\}$, and transforming four of the variables $(b_{i1},b_{i2},b_{i3},b_{i4})$ in the following manner ($d_i\in \mathbb{R}$): 
    \[x\mapsto \upsilon(b_{i1},b_{i2},b_{i3},b_{i4})+d_i(x-\upsilon(b_{i1},b_{i2},b_{i3},b_{i4})).\]
    A \textit{column transformation} is defined similarly, but the variables $(b_{1i},b_{2i},b_{3i},b_{4i})$ are used instead.
\end{definition}

\begin{theorem}
    Both row and column transformations preserve the $\upsilon_2$ value.
\end{theorem}

\begin{proof}
    Consider an arbitrary row transformation on $(b_{i1},b_{i2},b_{i3},b_{i4})$. Note that one of the terms in $\upsilon_2$ is $\upsilon(b_{i1},b_{i2},b_{i3},b_{i4})$, so if we can prove that this does not change, then we are done. Since $\upsilon$ is a weighted average, the transformation $x\mapsto x-\upsilon(b_{i1},b_{i2},b_{i3},b_{i4})$ turns the $\upsilon$ value of these 4 numbers to 0. Furthermore, the transformation $x\mapsto d_ix$ will multiply the $\upsilon$ value by $d_i$, still leaving it at 0. Finally, the transformation $x\mapsto x+\upsilon(b_{i1},b_{i2},b_{i3},b_{i4})$ makes the $\upsilon$ value return to its original value, as desired. For column transformations, apply this same logic in the column-based version (using $\upsilon(b_{1i},b_{2i},b_{3i},b_{4i})$) of $\upsilon_2$, as weighted average $\upsilon$ are row-column equivalent by \ref{weightavg}.
\end{proof}

By the above theorem, we can change Bob's transition values for the $N=2$ case in 8 ways: applying row transformations on $(b_{i1},b_{i2},b_{i3},b_{i4})$ and applying column transformations on $(b_{1i},b_{2i},b_{3i},b_{4i})$ for $i\in \{1,2,3,4\}$. 

Now, consider the scenario for a general $N$ and a set of Bob transition values $\{b_{i_1i_2\dots i_N}\}$. Additionally, define $\upsilon_N$ in a recursive manner based on $\upsilon_{N-1}$ (with $\upsilon_1:= \upsilon$):
\[\upsilon_N(\{b_{i_1i_2\dots i_N}\}):=\upsilon(\upsilon_{N-1}(\{b_{1i_2\dots i_{N}}\}),\upsilon_{N-1}(\{b_{2i_2\dots i_{N}}\}),\upsilon_{N-1}(\{b_{3i_2\dots i_{N}}\}),\upsilon_{N-1}(\{b_{4i_2\dots i_{N}}\})).\]
Note that the $N$-dimensional analogue of \ref{weightavg} is true, and so we can construct row-column-type transformations for each of the $N$ dimensions:
\[x\mapsto \upsilon(b_{1i_2\dots i_N},b_{2i_2\dots i_N},b_{3i_2\dots i_N},b_{4i_2\dots i_N})+d_i(x-\upsilon(b_{1i_2\dots i_N},b_{2i_2\dots i_N},b_{3i_2\dots i_N},b_{4i_2\dots i_N})).\]
The number of transformations of this type is $N\cdot 4^{N-1}$ ($N$ choices for which coordinate to vary from 1 to 4, and 4 options for each of the $N-1$ fixed coordinates). However, note that for $N\geq 4$, we have $N\cdot 4^{N-1}\geq 4^N$, the number of variables in Bob's transition map, giving us the following theorem:

\begin{theorem}\label{lindep}
    For $N\geq 4$, if there exists a family of weighted average ITV $\upsilon, \upsilon_2, \upsilon_3,\dots$ as defined previously, then the row-column-type transformations are linearly dependent of each other (i.e. any non-trivial row-column-type transformation affecting the row/column $T$ can be constructed from a composition of row-column-type transformations that do not directly act on $T$).
\end{theorem}

Due to the arbitrary nature of the $N\geq 4$ constraint in the statement of the theorem, we conjecture the following:

\begin{conjecture}\label{lindepsmall}
    Theorem \ref{lindep} holds true for $N=2$ and $N=3$.
\end{conjecture}

Finally, an alternative approach to resolve the case $N=1$ is to fix $P_{\overline{Z}|Z}(\overline{1},0)=\alpha$ and $P_{\overline{Z}|Z}(\overline{0},1)=\beta$, as well as $a_0,b_0,c_0,\dots,h_0,a_1,b_1,c_1,\dots, h_1$ (Eve's transition probabilities defined in tables \ref{Fig1} and \ref{Fig2}) at $\frac{1}2$, and then vary some subset of $\{a_0,b_0,c_0,\dots g_0,h_0\}$ and $\{a_1,b_1,c_1,\dots g_1,h_1\}$ to allow Eve to satisfy the independence condition.
We focus on the case $\overline{Z}=\overline{0}$.

Suppose that the channels that Bob use to binarize are $P_{\overline{Y}|Y}(\overline{0},0)=y_0$, $P_{\overline{Y}|Y}(\overline{0},1)=y_1$, $P_{\overline{Y}|Y}(\overline{0},2)=y_2$, and $P_{\overline{Y}|Y}(\overline{0},3)=y_3$, and define $x_0$, $x_1$, $x_2$, and $x_3$ for Alice's binarization similarly. Note that 
\[\frac{P\left(\overline{X}=\overline{0},\overline{Y}=\overline{0},\overline{Z}=\overline{0}\right)}{P\left(\overline{X}=\overline{0},\overline{Y}=\overline{1},\overline{Z}=\overline{0}\right)}=\frac{(x_0+x_1)(y_0+y_1)+(x_0+x_1)(y_2+y_3)+(x_2+x_3)(y_0+y_1)+2(x_2y_2+x_3y_3)}{(\text{same as above with all $y_i$ replaced with $1-y_i$})}\]
\[=\frac{(x_0+x_1+x_2+x_3)(y_0+y_1+y_2+y_3)+(x_2-x_3)(y_2-y_3)}{(x_0+x_1+x_2+x_3)(4-y_0-y_1-y_2-y_3)-(x_2-x_3)(y_2-y_3)}\]
Note that in the original formulation of the problem, the desired equality of the independence equation is a collection of weighted averages of the $\{x_i\}$ and $\{y_i\}$. As a result, we can assume WLOG that $x_0+x_1+x_2+x_3=y_0+y_1+y_2+y_3=0$ (while in turn losing the probabilistic meaning behing these variables). This simplifies the above ratio to $-1$. Similarly, note that
\[\frac{P\left(\overline{X}=\overline{1},\overline{Y}=\overline{0},\overline{Z}=\overline{0}\right)}{P\left(\overline{X}=\overline{1},\overline{Y}=\overline{1},\overline{Z}=\overline{0}\right)}=\frac{(4-x_0-x_1-x_2-x_3)(y_0+y_1+y_2+y_3)-(x_2-x_3)(y_2-y_3)}{(4-x_0-x_1-x_2-x_3)(4-y_0-y_1-y_2-y_3)+(x_2-x_3)(y_2-y_3)}.\]
Using the above simplifications, the ratio reduces to $\frac{-\zeta}{16+\zeta}$, where $\zeta:=(x_2-x_3)(y_2-y_3)$. In order for $\overline{X}\perp\!\!\!\perp \overline{Y}|\overline{Z}=\overline{0}$, we must have these two fractions be equal. However, the equality $-1=\frac{-\zeta}{16+\zeta}$ fails to hold for any $\{x_i\}$ and $\{y_i\}$. As a result, some isolated perturbations must be done to some of the variables. One possibility is to take $a_0=\frac{1}2$ and alter it slightly by a small value. For some $\epsilon\approx 0$, the new values of the two ratios would be
\[\frac{P\left(\overline{X}=\overline{0},\overline{Y}=\overline{0},\overline{Z}=\overline{0}\right)}{P\left(\overline{X}=\overline{0},\overline{Y}=\overline{1},\overline{Z}=\overline{0}\right)}=\frac{\zeta+\epsilon x_0y_2}{-\zeta+\epsilon x_0(1-y_2)},\]
\[\frac{P\left(\overline{X}=\overline{1},\overline{Y}=\overline{0},\overline{Z}=\overline{0}\right)}{P\left(\overline{X}=\overline{1},\overline{Y}=\overline{1},\overline{Z}=\overline{0}\right)}=\frac{-\zeta+\epsilon (1-x_0)y_2}{16+\zeta+\epsilon (1-x_0)(1-y_2)}.\]
In order for these two fractions to be equal, we must have
\[(16+\zeta+\epsilon (1-x_0)(1-y_2))(\zeta+\epsilon x_0y_2)=(-\zeta+\epsilon x_0(1-y_2))(-\zeta+\epsilon (1-x_0)y_2)\implies\]
\[16\zeta+16\epsilon x_0y_2+\epsilon\zeta(x_0y_2+(1-x_0)(1-y_2))=-\zeta\epsilon(x_0(1-y_2)+(1-x_0)y_2)\implies\]
\[16\zeta+16\epsilon x_0y_2+\epsilon\zeta=0.\]
If $\zeta$ is sufficiently small, then an $\epsilon$ can be chosen to satisfy the given equation, giving us the following theorem:
\begin{theorem}\label{zetasmall}
    For $\overline{X}$ and $\overline{Y}$ with $\zeta:=(x_2-x_3)(y_2-y_3)$ sufficiently small, there exists a binarization $\overline{Z}$ such that $(\overline{X}\perp\!\!\!\perp \overline{Y})|\overline{Z}$.
\end{theorem}

In future work we hope to generalize this perturbation technique to solve the problem for more classes of channels $(\overline{X},\overline{Y})$.

\section{Conclusion}

In this paper, we have shown a new relation between two well-known information-theoretic quantities: the intrinsic information and the reduced intrinsic information. Namely, for a given $P_{XYZ}$, when the reduced intrinsic information of this distribution is 0, then so is the intrinsic information. This relation has many important ramifications for significant conjectures in information theory. For example, out of the two long-standing conjectures of the secret-key rate being equal to the reduced intrinsic information and the conjecture of bound secrecy\cite{RenWol03}, at least one of them must be incorrect. Another implication is that the reduced intrinsic information cannot be used to prove that a distribution is bound secret. Future work in this direction would be to develop an information-theoretic quantity which has the property that it is not necessarily equal to 0 if the intrinsic information is equal to 0, and use this property to demonstrate that a particular distribution is bound secret.

We have also made progress on a possible approach for showing that a bound secret distribution does exist, using the idea of binarization of random variables \cite{GiReWo02}. In particular, we have reduced bound secrecy to a problem that does not require the use of information-theoretic quantities to formulate, instead using only basic ideas from probability. We have made progress on proving this statement for the candidate distribution introduced in \cite{GiReWo02}, by creating an explicit construction for an information-erasing binarization. The construction makes generalizing the information-erasing binarization much easier compared to the previous non-constructive results.

Furthermore, we have also made progress on proving bound secrecy for a family of distributions introduced in \cite{RenWol03}. In particular, we show that binarizing $Y$ alone is not sufficient to create independence between Alice and Bob given Eve, suggesting the underlying difference between proving bound secrecy for this distribution and the candidate distribution introduced in \cite{GiReWo02}. Additionally, we provide evidence that only Z-shaped channels need to be considered when binarizing. We also provide additional promising approaches for proving bound secrecy for this family of distributions, such as considering a particular class of weighted average target values and the row-column-type transformations they induce, and perturbing a single variable in Eve maps to solve bound secrecy in the $N=1$ case for a particular class of binarizations.

\section{Acknowledgements}
We would like to thank the MIT PRIMES-USA program for the opportunity to conduct this research. We would also like to thank Peter Shor for suggesting this problem to us. We also acknowledge Stefan Wolf, Matthias Christandl, and Renato Renner for their helpful answers to our questions regarding their papers.

\bibliography{citations}
\bibliographystyle{abbrv}

\end{document}